\documentclass{article}
\usepackage{spconf,amsmath}
\usepackage{amsthm}
\usepackage{amsfonts}
\usepackage{amssymb}
\usepackage{mathrsfs}
\usepackage[english]{babel}
\usepackage{float}
\usepackage[toc,page]{appendix}
\usepackage[pdftex]{graphicx}
\usepackage{epstopdf}
\usepackage{amsmath}
\usepackage{color}
\usepackage{subcaption}

\numberwithin{equation}{section}

\newtheorem{theorem}{Theorem}[section]

\newtheorem{corollary}[theorem]{Corollary}

\ninept
 \usepackage{algorithm} 
  \usepackage{algpseudocode}

\title{ A LEAST SQUARES APPROACH FOR STABLE PHASE RETRIEVAL FROM SHORT-TIME FOURIER TRANSFORM MAGNITUDE}
\name{Tamir Bendory and Yonina C. Eldar \thanks{This work was funded by the European Union’s Horizon 2020 research
and innovation programme under grant agreement ERC-BNYQ, by the Israel Science Foundation under Grant no. 335/14, and by ICore: the Israeli Excellence Center ‘Circle of Light’.}}
\address{Department of Electrical Engineering, Technion -– Israel Institute of Technology, Haifa, Israel.}
\begin{document}
\maketitle

\begin{abstract}
We address the problem of recovering a signal (up to global phase)  from its short-time Fourier transform (STFT) magnitude measurements. This problem arises in several applications, including optical imaging and speech processing. In this paper we suggest three interrelated algorithms. The first algorithm estimates the signal efficiently from noisy measurements by solving a simple least-squares (LS) problem. In contrast to previously proposed algorithms, the LS approach has stability guarantees and does not require any prior knowledge on the sought signal. However, the recovery is guaranteed under relatively strong restrictions on the STFT window. The second approach is guaranteed to recover a non-vanishing signal efficiently from noise-free measurements, under very moderate conditions on the STFT window. Finally, the third method estimates the signal robustly from noisy measurements by solving a  semi-definite program (SDP). The proposed SDP algorithm contains an inherent trade-off between its robustness and the restrictions on the STFT windows that can be used.
\end{abstract}

\section{Introduction} \label{sec:intro}
The problem of recovering a signal from its Fourier transform magnitude arises in many areas in engineering and science, such as optics, X-ray crystallography, speech recognition and blind channel estimation \cite{harrison1993phase,walther1963question,millane1990phase,juang1993fundamentals,baykal2004blind}. 
This problem is called \emph{phase retrieval} and received  considerable attention recently, partly due to its strong connections with the fields of compressed sensing and sparse recovery \cite{candes2015phase,candes2015Wirtinger,shechtman2014gespar}. 
We refer the reader to a contemporary survey of the phase retrieval problem in \cite{shechtman2014phase}. 

Phase retrieval is inherently an ill-posed problem. Two main approaches have been suggested to overcome the ill-posedness. The first  builds upon prior knowledge on the signal's support, such as  sparsity or exact knowledge of a portion of the underlying signal \cite{gerchberg1972practical,fienup1982phase,shechtman2011sparsity,ranieri2013phase,jaganathan2013sparse}.   The second approach makes use of additional  measurements. This can be performed by structured illuminations and masks \cite{candes2014phase,gross2014improved} or by measuring the magnitude of the short-time Fourier transform (STFT). In \cite{eldar2015sparse}, it was demonstrated that for the same number of measurements, measuring the STFT magnitude can lead to better performance than measuring an oversampled discrete Fourier transform (DFT).

In this paper we consider the  recovery of a signal from its STFT magnitude. This problem arises in several applications in optics  and speech processing \cite{nawab1983signal,griffin1984signal,trebino2002frequency}. In \cite{eldar2015sparse}, the authors prove that non-vanishing signals (namely, signals whose entries are all non-zero) can be recovered by an algebraic method, under mild  conditions on the STFT window. However, the algorithm is highly sensitive to noise.
In \cite{jaganathan2015stft}, the authors prove that most non-vanishing one-dimensional signals can be recovered by a semi-definite program (SDP) under mild conditions on the STFT window. They also
 demonstrate by numerical experiments that the algorithm is robust to noise. Two additional algorithms that often work well in practice are the iterative Griffin-Lim algorithm (GLA) \cite{griffin1984signal} and STFT-GESPAR (for sparse signals) \cite{shechtman2014gespar,eldar2015sparse}, however neither have recovery guarantees.

In this paper, we suggest three algorithms for the recovery of a signal from its STFT magnitude together with recovery guarantees. All of the proposed techniques are based on observing the DFT of the STFT magnitude. Algorithm \ref{alg:LS} is based on a simple least-squares (LS) approach. In contrast to previous works, the LS algorithm has stability guarantees. Furthermore, it can be implemented efficiently and does not require any prior knowledge on the sought signal, whereas previous methods require the knowledge of a portion of the  signal \cite{nawab1983signal} or that the signal is non-vanishing \cite{eldar2015sparse,jaganathan2015stft}. However, these appealing properties hold under certain restrictions on the STFT windows. 

To alleviate the constraints on the window, we suggest an efficient algebraic technique, Algorithm \ref{alg}, to recover non-vanishing signals. The algorithm shares similarities with  the method suggested in \cite{eldar2015sparse}, however it has  slightly less restrictions on the STFT window. Both Algorithms \ref{alg:LS} and \ref{alg} can be extended naturally to higher dimensions. 

Our final approach exploits a larger set of measurements in an algorithmically more complicated way using an SDP, to robustly recover  arbitrary signals. Our SDP formulation imposes different constraints than the one analyzed in \cite{jaganathan2015stft}. Simulations show that the algorithm performs well in noisy environments. It contains the flexibility to improve the signal's estimation (by relying on more observations) at the cost of more restrictions on the STFT window. 

The rest of the paper is organized as follows. Section \ref{sec:problem_formualtion} formulates the problem of recovering a signal from its STFT magnitude. Section \ref{sec:pr} presents our  algorithms, proves some of their  properties and discusses  connections with previous works. Section \ref{sec:numerics} shows numerical simulations, validating our theoretical findings.  Section \ref{sec:conclusion} concludes the work and  suggests future research directions. 
 
\section{Problem Formulation} \label{sec:problem_formualtion}

The  STFT of a signal  $\mathbf{x}\in\mathbb{C}^N$ is defined as 
\begin{equation} \label{eq:stft}
\mathbf{X}[m,k]=\sum_{n=0}^{N-1}\mathbf{x}[n] \mathbf{g}[mL-n]e^{-2\pi j kn/N},
\end{equation} 
where $k=0,\dots,N-1$, $m=0,\dots,\left \lceil 
\frac{N}{L} \right\rceil-1$,   
$\mathbf{g}[n]$ is the STFT window and $L$ determines the separation in time between adjacent sections.  In the sequel, we assume  that $\mathbf{x}$ and $\mathbf{g}$
are periodically extended over the boundaries in (\ref{eq:stft}). The STFT can be interpreted as the Fourier transform of the signal $\mathbf{x}$,  multiplied by a sliding window $\mathbf{g}$. 
In this work we consider $L=1$. The analysis  of $L>1$ will be left for future work.

%

We assume that the information we have on the signal is the magnitude of its STFT  
\begin{equation} \label{eq:meas}
\mathbf{{Y}}[m,k]=\left\vert \mathbf{X}[m,k]\right\vert^2+\mathbf{\eta}[m,k],
\end{equation}
where $\mathbf{\eta}[m,k]$ is an additive noise. Our goal is to recover $\mathbf{x}$ given $\mathbf{Y}$. Clearly, any signal $\mathbf{z}$ and $\mathbf{z}e^{j\phi}$, for some angle $\phi$, have the same  STFT magnitude. Therefore, the global phase $\phi$ cannot be determined by any method, and we say that the recovery is \emph{up to global phase}.

In \cite{jaganathan2015stft}, an SDP algorithm was suggested to estimate a non-vanishing signal  $\mathbf{x}$ from $\mathbf{Y}$. However, the high computational complexity of the SDP problem prevents its practical use (see for instance \cite{candes2015Wirtinger}). Additionally, while the SDP algorithm of \cite{jaganathan2015stft} seems to perform well in a noisy environment, it does not have stability guarantees. In this paper we suggest a simple and efficient algorithm to recover an arbitrary signal based on a LS formulation with stability guarantees by observing the DFT of the STFT magnitude.

Before proceeding, we introduce some notation. Let $\bar{\mathbf{z}}$ and $\mathbf{z}^*$ be the conjugate and the transpose conjugate of a vector $\mathbf{z}$, respectively. We say that $\mathbf{z}$ is non-vanishing if $\mathbf{z}[n]\neq 0$ for all $n$.
Let $\mathbf{F}$ and  $\mathbf{F}^*$ be the DFT and inverse DFT matrices. The set $\mathcal{H}_N$ denotes  all Hermitian matrices of size $N \times N$, and  $\vert \Lambda\vert$ the cardinality of the set $\Lambda$. For $k=0,1,\dots,N-1$, we denote by diag$(\mathbf{X},k)\in\mathbb{C}^N$ the $k^{th}$ circular diagonal of a matrix $\mathbf{X}$, namely,  the entries $(i,(i+k)\mbox{mod}N),\thinspace i=0,\dots,N-1$ of  $\mathbf{X}$.


\section{Phase Retrieval from STFT Magnitude} \label{sec:pr} 

\subsection{Stable Recovery via Least-Squares Estimation}  \label{sec:ls}

We begin by considering the DFT of $\mathbf{{Y}}[m,k]$ with respect to the second variable:
\begin{equation} \label{eq:dft}
\mathbf{Z}[m,\ell]=\sum_{k=0}^{N-1}\mathbf{Y}[m,k]e^{-2\pi jk\ell/N}.
\end{equation}
In the noiseless case we get
\begin{equation*}
\begin{split}
&\mathbf{Z}[m,\ell]=\sum_{k=0}^{N-1}\left\vert \mathbf{X}[m,k]\right\vert^2e^{-2\pi jk\ell/N} \\
=&\sum_{n_1,n_2=0}^{N-1}\mathbf{x}\left[n_1\right]\overline{\mathbf{x}\left[n_2\right]}\mathbf{g}\left[m-n_1\right]\overline{\mathbf{g}\left[m-n_2\right]} \sum_{k=0}^{N-1}e^{-2\pi j (n_1-n_2+\ell)k/N} \\
=&N\sum_{n=0}^{N-1}\mathbf{x}\left[n\right]\overline{\mathbf{x}\left[n+\ell\right ] }\mathbf{g}\left[(m-n)\mbox{mod}N\right]\overline{\mathbf{g}\left[(m-n-\ell)\mbox{mod}N\right]},
\end{split}
\end{equation*}
where the last equality follows from the identity $\sum_{k=0}^{N-1}e^{-2\pi j kn/N}=N\delta[n\mbox{mod}N]$.
Consequently, for any fixed $\ell=0,\dots,N-1$, we obtain the linear system of equations 
\begin{equation} \label{eq:yl}
\frac{1}{N}\mathbf{z}_\ell=\mathbf{G}_\ell\mathbf{x}_\ell,
\end{equation}
where $\mathbf{z}_\ell:=\left\{ \mathbf{Z}[m,\ell]\right\}_{m=0}^{N-1}$, $\mathbf{x}_\ell:=\left\{ \mathbf{x}\left[n\right] \overline{\mathbf{x}\left[n+\ell\right ]}\right\}_{n=0}^{N-1}$ and $\mathbf{G}_\ell$ is an $N \times N$ circulant matrix with  first column  given by $ \mathbf{g}[m]\overline{\mathbf{g}[(m-\ell) \mbox{mod} N]} $. 

Relation (\ref{eq:yl}) will be the basis for the three algorithms we propose below. Let $\mathbf{X}=\mathbf{x}\mathbf{x}^*$ and observe that diag$(\mathbf{X},\ell)=\mathbf{x}_\ell$. Therefore, in the noiseless case, if the matrices $\mathbf{G}_\ell$ are invertible for all $\ell=0,\dots,N-1$, then we can compute  $\mathbf{x}_\ell=\frac{1}{N}\mathbf{G}_\ell^{-1}\mathbf{z}_\ell$. Since $\mathbf{G}_\ell$ is a circulant matrix, it can be diagonalized by $\mathbf{G}_\ell=\mathbf{F}^*\mathbf{\Sigma_\ell}\mathbf{F}$, where $\mathbf{\Sigma_\ell}$ is a diagonal matrix whose entries are given by the DFT of the first column of $\mathbf{G}_\ell$. Thus, $\mathbf{x}_\ell=\frac{1}{N}\mathbf{F}^*\mathbf{\Sigma_\ell}^{-1}\mathbf{F}\mathbf{z}_\ell$. This step can be computed efficiently using FFT.  The set of vectors $\mathbf{x}_\ell, \thinspace \ell=0,\dots,N-1$ determines the rank one matrix $\mathbf{X}$.

In the presence of noise, we aim to solve the problem 
\begin{equation*} 
\min_{\mathbf{X}\in\mathbb{C}^{N\times N}} \sum_{\ell=0}^{N-1}\Vert \mathbf{z}_\ell-\mathbf{G}_\ell\mbox{diag}(\mathbf{X},\ell)\Vert_2^2 \quad \mbox{subject to} \quad \mbox{rank}(\mathbf{X})=1. 
\end{equation*}
In our first approach, we first ignore the rank-one constraint and solve a LS problem. If the matrices $\mathbf{G}_\ell$ are invertible,  then the solution is  given by 
\begin{equation*}
\mbox{diag}(\mathbf{X},\ell)=\frac{1}{N}\mathbf{G}_\ell^{-1}\mathbf{z}_\ell=\frac{1}{N}\mathbf{F}^*\mathbf{\Sigma_\ell}^{-1}\mathbf{F}\mathbf{z}_\ell.
\end{equation*}
We then determine $\mathbf{x}$ by computing the best rank one approximation of  $\mathbf{X}$, which is given by the eigenvector of $\mathbf{X}$ associated with the largest eigenvalue. This approach is summarized in Algorithm \ref{alg:LS}.

In a noise-free environment, Algorithm \ref{alg:LS} results in exact recovery under the following conditions:
\begin{theorem} \label{th:LS}
Any complex signal can be recovered (up to global phase) from its STFT magnitude with $L=1$  by Algorithm \ref{alg:LS} if the DFT of $\mathbf{g}[m]\mathbf{g}[(m-\ell)\mbox{mod}N]$ is non-vanishing for all $\ell=0,\dots,N-1$.
\end{theorem}
\begin{proof}
In the noiseless case, Algorithm \ref{alg:LS} recovers the $\mathbf{x}$ (up to global phase) if  $\mathbf{G}_\ell$ is invertible for all $\ell=0,\dots,N-1$. Since $\mathbf{G}_\ell$ is a circulant matrix, it is invertible if the DFT of its first column, given by $\mathbf{g}[m]\mathbf{g}[(m-\ell)\mbox{mod}N]$, is non-vanishing.
\end{proof}

In contrast to the algorithms suggested in \cite{eldar2015sparse,jaganathan2015stft}, Algorithm \ref{alg:LS} is guaranteed to achieve a stable estimation in the presence of noise as the solution of a LS problem as long as the matrices $\mathbf{G}_\ell$ are invertible, without any prior knowledge on the sought signal.  The algorithm is efficient due to the circular structure of the matrices $\mathbf{G}_\ell$ that can be inverted by FFT. 

The main drawback of Algorithm \ref{alg:LS} is the requirement for invertibility of $\mathbf{G}_\ell$ for all $\ell=0,1,\dots,N-1$. A common choice of the STFT window is a rectangular window $\mathbf{g}[n]=\mathbf{1}_{[0,W-1]}$, with ones in the interval $[0,W-1]$ and zero elsewhere. Conditions for invertibility in this case are given by the following corollary: 
\begin{corollary} \label{cor:ls}
Algorithm  \ref{alg:LS} recovers $\mathbf{x}$ (up to global phase) from its STFT magnitude with a rectangular window  $\mathbf{1}_{[0,W-1]}$ and $L=1$, if:
\begin{enumerate}
\item  $N/2<W<N$.
\item $N$ and the set $\{2W-N,\dots, W\}$ are coprime numbers. This always holds if $N$ is a prime number.
\end{enumerate}  
\end{corollary}
\begin{proof}
If $W=N$, then the DFT of $\mathbf{g}[m]\mathbf{g}[(m-\ell)\mbox{mod}N]$ is a Delta function for all $\ell$ and thus vanishes. 
If $W\leq N/2$ then there exists $\ell$ for which $\mathbf{g}[m]\mathbf{g}[(m-\ell)\mbox{mod}N]=0$ and thus $\mathbf{G}_\ell$ is not invertible for all $\ell$. If $W>N/2$ then  $\mathbf{g}[m]\mathbf{g}[(m-\ell)\mbox{mod}N]$ is a rectangular window of length ranging between $2W-N$ and $W$ (as a function of $\ell)$. Therefore, the DFT of the first column of $\mathbf{G}_\ell$ is a Dirichlet kernel (modulated by an exponential) with width ranging between $2W-N$ and $W$. If $N$ and all ${2W-N,\dots,W}$ are coprime then the Dirichlet kernel will not vanish. Clearly, if $N$ is a prime number then this statement holds.  
\end{proof}

\begin{algorithm} 
\caption{Signal Recovery from its noisy STFT magnitude using LS}
\begin{algorithmic}
\item \textbf{Input:}   The noisy magnitude of the signal STFT (\ref{eq:meas}). 
\item \textbf{Output:} Estimation of the sought signal $\hat{\mathbf{x}}$. 
\begin{enumerate}
\item Compute $\mathbf{Z}[m,\ell]$ according to (\ref{eq:dft}).
\item  For each $\ell=0,\dots,N-1$, compute $\mathbf{x}_\ell=\frac{1}{N}\mathbf{F}^*\mathbf{\Sigma}_\ell^{-1}\mathbf{F} \mathbf{z}_\ell$, where $\mathbf{\Sigma}_\ell$ is a diagonal matrix with the entries $\mathbf{F}\left(\mathbf{g}[n]\mathbf{g}[(n-\ell)\mbox{mod}N]\right)$.
\item Construct a matrix $\mathbf{X}\in\mathbb{C}^{N \times N}$ by  $\mbox{diag}(\mathbf{X},\ell)=\mathbf{x}_\ell$ for all $\ell=0,\dots,N-1$.
\item  Compute  $\hat{\mathbf{x}}=\sqrt{\lambda_{\max}}u_{\max}$, where $\lambda_{\max}$ is the maximal eigenvalue of $\mathbf{X}$ and $u_{\max}$ is the associated eigenvector.    
\end{enumerate}
\end{algorithmic}
 \label{alg:LS}

\end{algorithm}

\subsection{ Algebriac Recovery in a Noise-Free Environment}\label{sec:noisefree}
To reduce the requirements on the window, suppose that the information we have is exactly $\left\vert \mathbf{X}[m,k]\right\vert^2$ without any additional noise. 
Theorem \ref{th:main_algebriac} states that one can recover efficiently any non-vanishing signal $\mathbf{x}\in\mathbb{C}^N$ from its STFT magnitude via Algorithm \ref{alg} if the  STFT window $\mathbf{g}$ satisfies  moderate conditions.  Compared to Theorem 1 in \cite{eldar2015sparse}, Theorem \ref{th:main_algebriac} does not have explicit restrictions on the size of the window, however it has an additional requirement that the DFT of $\mathbf{g}[m]\overline{\mathbf{g}[(m-\ell)\mbox{mod}N])}$ for $\ell=0,1,$ is non-vanishing. 
Theorem 4.2 in \cite{jaganathan2015stft} has similar recovery guarantees, however it requires solving an SDP, which has high computational complexity.

\begin{algorithm} 
\caption{Signal Recovery from its noise-free STFT magnitude}
\begin{algorithmic}
\item \textbf{Input:}  The magnitude of the signal STFT $\left\vert \mathbf{X}[m,k]\right\vert^2$. 
\item \textbf{Output:} Estimate of the sought signal $\hat{\mathbf{x}}$. 
\begin{enumerate}
\item Compute $\mathbf{Z}[m,\ell]$ according to (\ref{eq:dft}).
\item  For $\ell=0,1,$ compute $\mathbf{x}_\ell=\frac{1}{N}\mathbf{F}^*\mathbf{\Sigma}_\ell^{-1}\mathbf{F} \mathbf{z}_\ell$, where $\mathbf{\Sigma}_\ell$ is a diagonal matrix with entries $\mathbf{F}_N\left(\mathbf{g}[n]\overline{\mathbf{g}[(n-\ell)\mbox{mod}N])}\right)$  and set $\hat{\mathbf{x}}[0]=\sqrt{\mathbf{x}_0[0]}$.
\item Compute recursively $\overline{\hat{\mathbf{x}}[n+1]}=\frac{\mathbf{x}_1[n]}{\hat{\mathbf{x}}[n]}$.
\end{enumerate}
\end{algorithmic}
 \label{alg}

\end{algorithm}

\begin{theorem} \label{th:main_algebriac}
Any non-vanishing complex signal can be recovered (up to global phase) from $\left\vert \mathbf{X}[m,k]\right\vert^2$ with $L=1$  by Algorithm \ref{alg} if the DFT  of $\mathbf{g}[n]\overline{\mathbf{g}[(n-\ell)\mbox{mod}N]}$ is non-vanishing for  $\ell=0,1$.
\end{theorem}

\begin{proof} 
According to (\ref{eq:yl}), if $\mathbf{G}_\ell$ is invertible then we get the linear system of equations 
\begin{equation*}
\frac{1}{N}\mathbf{z}_\ell=\mathbf{G}_\ell\mathbf{x}_\ell \Leftrightarrow  \mathbf{x}_\ell=\frac{1}{N}\mathbf{F}^*\mathbf{\Sigma}_\ell^{-1}\mathbf{F}\mathbf{z}_\ell,
\end{equation*}
where ${\Sigma}_\ell$ is a diagonal matrix with entries  $F_N\left(\mathbf{g}[n]\overline{\mathbf{g}[(n-\ell) \mbox{mod} N]}\right)$.
By assumption, the DFT of $ \vert \mathbf{g}[n]\vert ^2$  is not equal to zero and thus we can compute $\mathbf{x}_0[n]=\vert \mathbf{x}[n]\vert^2$. If the signal is known to be non-negative, then the algorithm is completed.
By further assumption,  we can compute $\mathbf{x}_1[n]=\mathbf{x}[n]\overline{\mathbf{x}[n+1]}$. Since we cannot recover the global phase, we arbitrary set $\mathbf{x}[0]=\vert \mathbf{x}[0]\vert$ which is known and equal to $\sqrt{\mathbf{x}_0[0]}$. Now, we can compute the rest of the entries recursively according to step 3.
\end{proof}
%

Algorithm \ref{alg} alleviates the requirements of the Theorem \ref{th:LS} on the window. The following corollary demonstrates this  on a rectangular window. The proof is similar to the proof of Corollary \ref{cor:ls} and thus omitted.   
\begin{corollary} \label{cor}
Algorithm  \ref{alg} recovers a non-vanishing $\mathbf{x}$ (up to global phase) with a rectangular window  $\mathbf{g}[n]=\mathbf{1}_{[0,W-1]}$ and $L=1$, if:
\begin{enumerate}
\item  $2\leq W\leq N-1$.
\item $N$ and $W-\ell$ are coprime for $\ell=0,1$.
\end{enumerate}  
\end{corollary}

%


\subsection{Extension to High-Dimensional Signals} \label{sec:2d}
We briefly describe the extension of Algorithms \ref{alg:LS} and \ref{alg} to bivariate signals. The algorithms can be extended to higher dimensions by the same methodology.  The two-dimensional (2D) STFT with respect to a bivariate window $\mathbf{g}$ is given by (for $L=1$)
\begin{equation*}
\begin{split}
\mathbf{X}\left[ m_1,m_2 ; k_1,k_2 \right]=&\sum_{n_1,n_2=0}^{N-1}\mathbf{x}\left[ n_1,n_2\right]\mathbf{g}\left[ m_1-n_1,m_2-n_2\right]\\
&e^{-2\pi j \left( k_1n_1+k_2n_2\right)/N}.
\end{split}
\end{equation*}
In the same manner as in (\ref{eq:dft}), we take the 2D DFT of $\vert \mathbf{X}\left[ m_1,m_2 ; k_1,k_2 \right]\vert^2 $ with respect to $k_1$ and $k_2$ and get, \footnote{All the indices in this section should be interpreted as modolu $N$.}
\begin{equation*}
\begin{split}
&\frac{1}{N^2}\mathbf{Z}[m_1,m_2 ; \ell_1,\ell_2] \\ &=\frac{1}{N^2}\sum_{k_1,k_2=0}^{N-1}\vert \mathbf{X}\left[ m_1,m_2 ; k_1,k_2 \right]\vert ^2 e^{-2\pi j \left(k_1\ell_1+k_2\ell_2 \right)/N} \\
&=\sum_{n_1,n_2=0}^{N-1}\mathbf{x}\left[n_1,n_2 \right]\overline{\mathbf{x}\left[n_1+\ell_1,n_2+\ell_2 \right]} \\
& \mathbf{g}\left[ m_1-n_1,m_2-n_2\right]\overline{\mathbf{g}\left[ m_1-n_1-\ell_1,m_2-n_2-\ell_2\right]}.
\end{split}
\end{equation*}
For a fixed $(\ell_1,\ell_2)$, we obtain the linear system of equations
\begin{equation}
\begin{split}
&\frac{1}{N^2}\mathbf{Z}[m_1,m_2 ; \ell_1,\ell_2] =\mathbf{G}_{\ell_1,\ell_2}\mathbf{x_{\ell_1,\ell_2}}, 
\end{split}
\end{equation}
where $\mathbf{x_{\ell_1,\ell_2}}=\mathbf{x}\left[n_1,n_2 \right]\overline{\mathbf{x}\left[n_1+\ell_1,n_2+\ell_2 \right]}$ and $\mathbf{G}_{\ell_1,\ell_2}$ is a block circulant matrix  with circulant blocks. Such a matrix is diagonalized by the 2D DFT matrix and the eigenvalues are given by the 2D DFT of the first column of $\mathbf{G}_{\ell_1,\ell_2}$. Therefore, if the 2D DFT of the first column of $\mathbf{G}_{\ell_1,\ell_2}$ is non-vanishing then the system is invertible and one can recover  $\mathbf{x_{\ell_1,\ell_2}}$ efficiently. The rest of the procedure is a simple generalization of the one-dimensional cases presented in Algorithms \ref{alg:LS} and \ref{alg}.

\subsection{Stable Recovery via Semidefinite Programming } \label{sec:sdp}
Algorithm \ref{alg} uses merely two columns of the matrix $\mathbf{Y}[m,\ell]$ (\ref{eq:dft}). Algorithm \ref{alg:sdp} suggests to exploit a larger set of columns of $\mathbf{Y}[m,\ell]$, to achieve robust estimation in the presence of noise.

Let $\Lambda\subseteq\{0,1,\dots,N-1\}$ and let $\vert\Lambda\vert$ be its cardinality. Algorithm \ref{alg:sdp} suggests to estimate the signal from noisy observations by minimizing the trace of a matrix among all positive semidefinite Hermitian matrices which satisfy the $\vert \Lambda\vert$  constraints on the matrix diagonals. It is well-known that minimizing the trace of an Hermitian matrix promotes a low-rank solution. 

We expect that choosing a larger set of observations $\vert \Lambda\vert$ will lead to a better estimation of $\mathbf{x}$. This is true if the matrices  $\mathbf{G}_\ell$ are invertible for all $\ell\in\Lambda$. Therefore, Algorithm \ref{alg:sdp} enables the flexibility  to choose the number of observation we use depending on the STFT window. For instance, if we consider the rectangular window $\mathbf{1}_{[0,W-1]}$ for some $W<N/2$ and the set $\Lambda=\{0,1,\dots,\lambda-1\}$, then the matrices $\mathbf{G}_\ell$ are invertible only if $\lambda\leq W$ and that $N$ and $W-\ell$ are coprime for all $\ell\in \Lambda$. If this is the situation, we expect that increasing $ \lambda $ (up to $W$) will result in a better estimation. The numerical results in Section \ref{sec:numerics} corroborate this conclusion.

 \begin{algorithm} 
\caption{Recovering a signal from its noisy STFT magnitude using SDP}
\begin{algorithmic}
\item \textbf{Input:}  The noisy magnitude of the signal STFT (\ref{eq:meas}). 
\item \textbf{Output:} Estimate of the sought signal $\hat{\mathbf{x}}$. 
\begin{enumerate}
\item Compute $\mathbf{Z}[m,\ell]$ according to (\ref{eq:dft}).
\item Choose a set $\Lambda\subseteq\{0,\dots,N-1\}$  and solve : 
\begin{equation} \label{eq:sdp}
\begin{split}
&{\mathbf{X}}=\min_{\widetilde{\mathbf{X}}\in\mathcal{H}_N} \mbox{trace}(\widetilde{\mathbf{X}}) \quad \mbox{subject to } \quad\widetilde{\mathbf{X}}\geq 0, \\ & \Vert \mathbf{z}_\ell-\mathbf{G}_\ell \mbox{diag}(\widetilde{\mathbf{X}},\ell)\Vert_2\leq \eta_\ell, \quad \forall\ell\in\Lambda,
\end{split}
\end{equation}
where $\eta_\ell $ is the noise level associated with the $\ell^{th}$ column of $\mathbf{Z}[m,\ell]$.
\item $\hat{\mathbf{x}}=\sqrt{\lambda_{\max}}u_{\max}$, where $\lambda_{\max}$ is the maximal eigenvalue of ${\mathbf{X}}$ and $u_{\max}$ is the associated eigenvector.
\end{enumerate}
\end{algorithmic}
 \label{alg:sdp}

\end{algorithm}


\section{Numerical Experiments} \label{sec:numerics}

This section is devoted to numerical experiments, examining the three algorithms presented above. In all experiments, both the real and the imaginary parts of the signals were drawn from an iid normal distribution. An additive noise with normal distribution was added with the appropriate variance, according to the desired signal to noise (SNR) ratio. The error was measured by $\frac{\Vert \mathbf{x}-\hat{\mathbf{x}}\Vert_2 }{\Vert \hat{\mathbf{x}}\Vert_2}$, where $\hat{\mathbf{x}}$ is the algorithm output.

We  compared Algorithm \ref{alg:LS} with the GLA \cite{griffin1984signal} and the noisy version of the SDP algorithm in \cite{jaganathan2015stft}. We solved the SDP using CVX \cite{grant2008cvx} and stopped the iterative GLA if the difference between consecutive iterations was less that $10^{-6}$ or after 500 iterations. Figure \ref{fig:LS} demonstrates that all three algorithms have similar performance for signals of length $N=23$ and Gaussian window $e^{-n^2/\sigma^2}$ with $\sigma=\left \lceil \frac{N}{2}\right \rceil$. However, the average running time of the GLA and the SDP of  \cite{jaganathan2015stft} was 27  and 2634 times greater than the running time of Algorithm \ref{alg:LS}, respectively. 
Table \ref{table} presents the recovery error of Algorithm \ref{alg} in a noise-free environment. We used a rectangular window with different lengths. As predicted by Theorem \ref{th:main_algebriac}, the recovery error is negligible.
Figure \ref{fig:sdp} compares the recovery error of Algorithm \ref{alg:sdp} with the SDP algorithm of \cite{jaganathan2015stft} for signals of length $N=23$ and a rectangular window with $W=5$ as a function of the SNR. The problems were solved by CVX. We chose $\Lambda=\{0,1,\dots,\lambda-1\}$ for $\lambda=3,5,23$ . 
As can be seen, Algorithm \ref{alg:sdp} is stable even for small values of $\lambda$. Additionally, increasing $ \lambda $ beyond $W$ does not reduce the recovery error significantly as predicted in Section \ref{sec:sdp}. The running time of  the SDP algorithm of \cite{jaganathan2015stft} was 3.5 and 12 times greater than the running time of Algorithm \ref{alg:sdp} with $\lambda=5,23$, respectively.  

\begin{figure} [h]

       \includegraphics[scale=0.4]{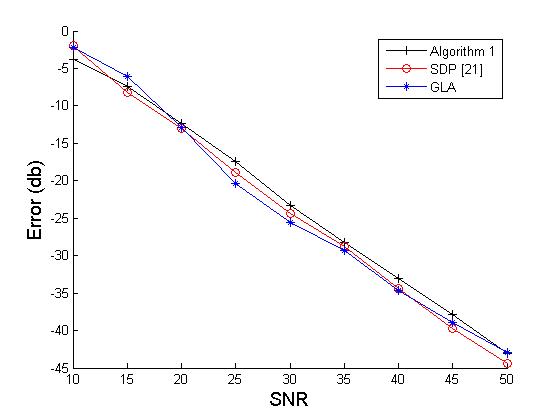}

\protect\caption{\label{fig:LS} The recovery error of Algorithm \ref{alg:LS}, the SDP method suggested in \cite{jaganathan2015stft} and the GLA \cite{griffin1984signal} as a function of the SNR.  Each experiment was conducted 10 times. }
\end{figure}

\begin{figure} [h]
\centering
       \includegraphics[scale=0.4]{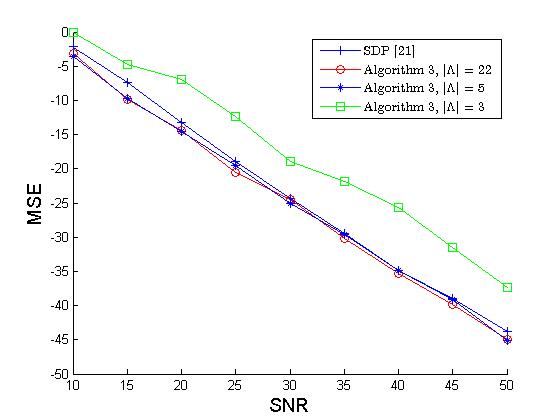}
\protect\caption{\label{fig:sdp} The recovery error of Algorithm \ref{alg:sdp} with a rectangular window of length $W=9$ and $N=23$ as a function of the SNR for different values of $\vert \Lambda\vert$. The experiment was conducted 10 times for each value.}
\end{figure}

\begin{table}[h]
\begin{center}
    \begin{tabular}{| l | l | l | l |} 
    \hline
     & $W=5$ & $W=23$ & $W=41$ \\ \hline
    Mean error & $3.52\times 10^{-12}$  & $6.84\times 10^{-12}$  & $1.13\times 10^{-11}$ \\ \hline
    Max error &  $1.46\times 10^{-12}$  & $3.05\times 10^{-12}$  & $7.02\times 10^{-11}$ \\ \hline    
    \end{tabular} 
\end{center}
    \caption{The recovery error of Algorithm \ref{alg} with a rectangular window of length $W$ and $N=211$ in a noise-free environment.  The experiment was conducted 100 times for each  window size. }      
    \label{table}
\end{table}

\section{Conclusion} \label{sec:conclusion}
In this work, we  suggested three interrelated algorithms  to recover a signal from its STFT magnitude. Each algorithm suits a different case. Algorithm \ref{alg:LS} is  efficient and robust to noise but restricts the potential STFT windows. Algorithm \ref{alg} is an efficient and exact algorithm in the noise-free environment. Algorithm \ref{alg:sdp} uses an SDP formulation and contains the flexibility to control its stability to noise at the cost of restrictions on the STFT windows. Proving its robustness and properties is the goal of future work.

All the results in this paper were focused on maximal overlapping in the STFT windows, i.e. $L=1$. Previous papers have indicated that this restriction can be alleviated. The analysis of the case $L>1$ requires incorporating new tools into the suggested framework and is under ongoing research.

 
\bibliographystyle{IEEEbib}
\bibliography{refs}

\end{document}